\documentclass[smallcondensed]{svjour3}[]
\usepackage{amsmath, amsfonts, amssymb, comment}
\usepackage{geometry}
\usepackage[small]{diagrams}

\newcommand{\rs}{\mathrm{rs}}
\newcommand{\GL}{\mathrm{GL}}
\newcommand{\PG}{\mathrm{PG}(\F_q^n)}
\newcommand{\Mat}{\mathrm{Mat}}
\newcommand{\F}{\mathbb{F}}
\newcommand{\G}{\mathcal{G}_q(k,n)}
\newcommand{\Vvs}{\mathcal{V}}
\newcommand{\Uvs}{\mathcal{U}}
\newcommand{\Cvs}{\mathcal{C}}

\smartqed

\begin{document}

\title{A Complete Characterization of Irreducible Cyclic Orbit Codes
  and their Pl\"ucker Embedding \thanks{Research partially supported
    by Swiss National Science Foundation Project no. 126948. A
    preliminary version of this paper was presented at the Seventh
    International Workshop on Coding and Cryptography (WCC) 2011.}} %
\author{Joachim Rosenthal \and Anna-Lena Trautmann}
\institute{Institute of Mathematics\\
  University of Zurich, Switzerland\\
  \email{www.math.uzh.ch/aa}}

\maketitle


\begin{abstract}
  Constant dimension codes are subsets of the finite Grassmann
  variety. The study of these codes is a central topic in random linear network coding theory.

  Orbit codes represent a subclass of constant dimension codes. They
  are defined as orbits of a subgroup of the general linear group on
  the Grassmannian.

  This paper gives a complete characterization of orbit codes that are
  generated by an irreducible cyclic group, i.e. a group having one
  generator that has no non-trivial invariant subspace. We show how
  some of the basic properties of these codes, the cardinality and the
  minimum distance, can be derived using the isomorphism of the vector
  space and the extension field. Furthermore, we investigate the
  Pl\"ucker embedding of these codes and show how the orbit structure
  is preserved in the embedding.


\end{abstract}

\section{Introduction}

In network coding one is looking at the transmission of information
through a directed graph with possibly several senders and several
receivers \cite{ah00}. One can increase the throughput by linearly
combining the information vectors at intermediate nodes of the
network.  If the underlying topology of the network is unknown we
speak about \textit{random linear network coding}. Since linear spaces
are invariant under linear combinations, they are what is needed as
codewords \cite{ko08}. It is helpful (e.g. for decoding) to constrain
oneself to subspaces of a fixed dimension, in which case we talk about
\emph{constant dimension codes}.

The set of all $k$-dimensional subspaces of a vector space $V$ 
is often referred to as the Grassmann variety (or simply Grassmannian) 
and denoted by 
$\mathcal{G}(k,V)$. {\em Constant dimension codes} are 
subsets of $\mathcal{G}(k,\F_{q}^n)$, where $\F_{q}$ is some finite field.

The general linear group $\GL(V)$ consisting of all invertible
transformations acts naturally on the Grassmannian $\mathcal{G}(k,V)$.
If $\mathfrak{G}\leq \GL(\F_{q}^n)$ is a subgroup then one has an
induced action of $\mathfrak{G}$ on the finite Grassmannian
$\mathcal{G}(k,\F_{q}^n)$. Orbits under the $\mathfrak{G}$-action are
called \textit{orbit codes} \cite{tr10p}. Orbit codes have useful
algebraic structure, e.g. for the computation of the distance of an
orbit code it is enough to compute the distance between the starting
point and any of its orbit elements. This is analogous to linear block
codes where the minimum distance of the code can be derived from the
weights of the non-zero code words.

Orbit codes can be classified according to the groups used to
construct the orbit. In this work we characterize orbit codes
generated by irreducible cyclic subgroups of the general linear group
and their Pl\"ucker embedding.

The paper is structured as follows: The second section gives some
preliminaries, first of random network coding and orbit codes. Then
some facts on irreducible polynomials are stated and the
representation of finite vector spaces as Galois extension fields is
explained in \ref{galois}. In part \ref{img} we introduce irreducible
matrix groups and give some properties, with a focus on the cyclic
ones.  The main body of the paper are Section \ref{icoc} and \ref{PC}.
In the former where we study the behavior of orbit codes generated by
these groups and compute the cardinality and minimum distances of
them. We begin by characterizing primitive orbit codes and then study
the non-primitive irreducible ones.  Section \ref{PC} deals with the
Pl\"ucker embedding of cyclic irreducible orbit codes. Finally we give
a conclusion and an outlook in Section \ref{conclusion}.

\section{Preliminaries}
\subsection{Random Network Codes}\label{oc}

Let $\mathbb{F}_q$ be the finite field with $q$ elements, where $q$ is
a prime power. For simplicity we will denote the Grassmannian
$\mathcal{G}(k,\F_{q}^n)$ by $\G$ and the general linear group, that
is the set of all invertible $n\times n$-matrices with entries in
$\F_{q}$, by $\GL_{n}$. Moreover, the set of all $k\times n$-matrices
over $\F_q$ is denoted by $\Mat_{k\times n}$.

Let $U\in \Mat_{k\times n}$ be a matrix of rank $k$ and
\[
\mathcal{U}=\rs (U):= \text{row space}(U)\in \G.
\]
One can notice that the row space is invariant under $\GL_k$-multiplication from
the left, i.e. for any $T\in \GL_k$
\[
\mathcal{U}=\rs(U)= \rs(T U).
\]
Thus, there are several matrices that represent a given subspace. A
unique representative of these matrices is the one in reduced row
echelon form.
Any $k\times n$-matrix can be transformed into reduced row echelon
form by a $T\in \GL_k$.

The set of all subspaces of $\F_q^n$, called the projective geometry of $\F_q^n$, is denoted by $\PG$. 
The \emph{subspace distance} is a metric on it, given by
\begin{align*}
  d_S(\mathcal{U},\mathcal{V}) =& \dim(\Uvs)+\dim(\Vvs) - 2\dim(\mathcal{U}\cap
  \mathcal{V})
\end{align*}
for any $\mathcal{U},\mathcal{V} \in \PG$. It is a suitable
distance for coding over the operator channel \cite{ko08}. 

A \textit{constant dimension code} $\mathcal{C}$ is simply a subset of the
Grassmannian $\G$. The minimum distance is defined
in the usual way. A code $\mathcal{C}\subset \G$ with minimum
distance $d_S(\mathcal{C})$ is called an $[n, d_S(\mathcal{C}), |\mathcal{C}|, k]$-code. Different constructions of constant dimension codes can be found in e.g. \cite{et08u,ko08p,ko08,ma08p,si08a,tr10p}.

In the case that $k$ divides $n$ one can construct \textit{spread
  codes} \cite{ma08p}, i.e. optimal codes with minimum distance $2k$.
These codes are optimal because they achieve the Singleton-like bound \cite{ko08},
which means they have $\frac{q^{n}-1}{q^{k}-1}$ elements.

Given $U\in \Mat_{k\times n}$ of rank $k$,
$\mathcal{U}\in \G$ its row space and $A\in \GL_n$, we
define
\[
\mathcal{U} A:=\rs(U A).
\]

Let $U,V\in \Mat_{k\times n}$ be matrices such that
$\rs(U)=\rs(V)$. Then one readily verifies that $\rs(U A)=\rs(V A)$
for any $A\in \GL_n$. The subspace distance is $\GL_{n}$-invariant, i.e. $d_{S}(\Uvs, \Vvs) = d_{S}(\Uvs A, \Vvs A) $ for $A\in \GL_n$.

This multiplication with $\GL_n$-matrices defines a
group operation from the right on the Grassmannian:

\[
\begin{array}{ccc}
  \G\times \GL_n& \longrightarrow &
  \G \\ 
  (\mathcal{U},A) & \longmapsto & \mathcal{U} A
\end{array} 
\]
Let $\mathcal{U}\in \G$ be fixed and $\mathfrak{G}$ a
subgroup of $\GL_n$. Then
\[
\mathcal{C}= \{\mathcal{U} A \mid A \in \mathfrak{G}\}
\]
is called an \emph{orbit code} \cite{tr10p}. It is well-known that 
\[
\G \cong \GL_n/Stab_{\GL_n}(\mathcal{U}),
\]
where $Stab_{\GL_n}(\mathcal{U}):=\{A\in \GL_n \mid \mathcal{U} A=\mathcal{U}\}$. 
There are different subgroups that generate the same orbit code.
An orbit code is called \emph{cyclic} if it can be defined by a cyclic
subgroup $\mathfrak{G} \leq \GL_n$.

\subsection{Irreducible Polynomials and Extension
  Fields}\label{galois}

Let us state some known facts on irreducible polynomials over finite
fields \cite[Lemmas 3.4 - 3.6]{li86} :

\begin{lemma}
  Let $p(x)$ be an irreducible polynomial over $\mathbb{F}_{q}$ of
  degree $n$, $p(0)\neq 0$ and $\alpha$ a root of it. Define the order
  of a polynomial $p(x) \in \F_{q}[x]$ with $p(0)\neq 0$ as the
  smallest integer $e$ for which $p(x)$ divides $x^e-1$.  Then
  \begin{enumerate}
  \item the order of $p(x)$ is equal to the order of $\alpha$ in
    $\F_{q^{n}}\backslash \{0\}$.
  \item the order of $p(x)$ divides $q^n-1$.
  \item $p(x)$ divides $x^c-1$ if, and only if the order of $p(x)$ divides $c$
    (where $c\in \mathbb{N}$).
  \end{enumerate}

\end{lemma}

There is an isomorphism between the vector space $\mathbb{F}_q^n$ and
the Galois extension field $\mathbb{F}_{q^n} \cong
\mathbb{F}_q[\alpha]$, for $\alpha$ a root of an irreducible
polynomial $p(x)$ of degree $n$ over $\mathbb{F}_q$. If in addition
$p(x)$ is primitive, then
\[\mathbb{F}_q[\alpha]\backslash\{0\}=\langle \alpha \rangle =
\{\alpha^i \mid i=0,...,q^n-2\}\] i.e.  $\alpha$ generates
multiplicatively the group of invertible elements of the extension
field.

\begin{lemma}\label{vsprim}
  If $k|n, c:=\frac{q^n-1}{q^k-1}$ and $\alpha$ a primitive element of
  $\mathbb{F}_{q^n}$, then the vector space generated by
  $1,\alpha^{c},..., \alpha^{(k-1)c}$ is equal to $\{\alpha^{ic} \mid
  i=0,...,q^{k}-2\}\cup\{0\}=\F_{q^k}$.
\end{lemma}
\begin{proof}
  Since $k|n$ it holds that $c\in \mathbb{N}$.  Moreover, it holds that
  $(\alpha^{c})^{q^{k}-1}= \alpha^{q^{n}-1}= 1$ and
  $(\alpha^c)^{q^{k}-2}= \alpha^{-c}\neq 1$, hence the order of
  $\alpha^{c}$ is $q^{k}-1$.  It is well-known that if $k$ divides $n$
  the field $\mathbb{F}_{q^n}$ has exactly one subfield
  $\mathbb{F}_{q^k}$. Thus the group generated by $\alpha^{c}$ has to
  be $\mathbb{F}_{q^k}\backslash \{0\}$, which again is isomorphic to
  $\mathbb{F}_q^{k}$ as a vector space.   \qed 
  \end{proof}



\subsection{Irreducible Matrix Groups}\label{img}

\begin{definition}
  \begin{enumerate}
  \item A matrix $A\in \GL_n$ is called \textsl{irreducible} if
    $\F_q^n$ contains no non-trivial $A$-invariant subspace, otherwise it is
    called \textsl{reducible}.
  \item A subgroup $\mathfrak{G}\leq \GL_n$ is called
    \textsl{irreducible} if $\F_q^n$ contains no
    $\mathfrak{G}$-invariant subspace, otherwise it is called
    \textsl{reducible}.
   \item  An orbit code $\mathcal{C} \subseteq \G$ is called 
    \textsl{irreducible} if $\mathcal{C}$ is the orbit 
    under the action of an irreducible group.
  \end{enumerate}
\end{definition}

A cyclic group is irreducible if and only if its generator matrix is
irreducible. Moreover, an invertible matrix is irreducible if and only
if its characteristic polynomial is irreducible.

\begin{example} 
  Over $\mathbb{F}_2$ the only irreducible polynomial of degree $2$ is
  $p(x)=x^2+x+1$. Since their characteristic polynomial has to be $p(x)$, the irreducible matrices in $\GL_2$ must have trace
  and determinant equal to $1$ and hence are
  \[\left(\begin{array}{cc} 0 & 1 \\ 1 & 1 \end{array}\right)
  \textnormal{ and } \left(\begin{array}{cc} 1 & 1 \\ 1 & 0
    \end{array}\right).\]
\end{example}

We can say even more about irreducible matrices
with the same characteristic polynomial.
For this, note that the definition of an irreducible matrix $G$
implies the existence of a {\em cyclic vector} $v\in \F_{q}^n$
having the property that
$$
\{ v,vG,vG^2,\ldots,vG^{n-1}\}
$$
forms a basis of $\F_{q}^n$. Let $S\in \GL_n$ be the basis transformation
which transforms the matrix $G$ into this new basis. Then it follows that
$$
SGS^{-1}=
\left( 
  \begin{array}{ccccc}
    0 & 1 & 0 &\ldots  &0\\
    0 & 0 & 1 &\ldots  &0\\
 \vdots&\vdots&\vdots&\ddots&\vdots\\
   0 & 0 & 0 &\ldots  &1\\
   -c_0 & -c_{1} & & \ldots  &-c_{n-1}
  \end{array}
\right).
$$
The matrix appearing on the right is said to be
in {\em companion form}. By convention we will use row
vectors $v\in \F_q^{n}$ and accordingly companion matrices where the coefficients of
the corresponding polynomials are in the last row (instead of the last
column).

One readily verifies that 
$$
p(x):=x^n+c_{n-1}x^{n-1}+\cdots+c_1x+c_0
$$
is the characteristic polynomial of both $G$ and $SGS^{-1}$.  It
follows that every irreducible matrix in $\GL_n$ is similar to the
companion matrix of its characteristic polynomial. Hence all
irreducible matrices with the same characteristic polynomial are
similar.

Furthermore, the order of
$G \in \GL_{n}$ is equal to the order of its characteristic polynomial.  Hence $\textrm{ord}(G)=q^{n}-1$ if and only if its characteristic polynomial is primitive.

The next fact is a well-known group theoretic result:

\begin{lemma}\cite[Theorem 1.15.]{li86}   In a finite cyclic group
  $\mathfrak{G}=\langle G \rangle$ of order $m$, the element $G^l$
  generates a subgroup of order $\frac{m}{\gcd(l,m)}$. Hence each
  element $G^l$ with $\gcd(l,m)=1$ is a generator of $\mathfrak{G}$.
\end{lemma}

\begin{lemma}\cite[Theorem 7]{ma11} 
   All irreducible cyclic groups generated by matrices with a
   characteristic polynomial of the same order are conjugate to each
   other.
\end{lemma}

\begin{example}
  Over $\mathbb{F}_2$ the irreducible polynomials of degree $4$ are
  $p_1(x)=x^4+x+1$, $p_2(x)=x^4+x^3+1$ and $p_3(x)=x^4+x^3+x^2+x+1$, where
  $\textrm{ord}(p_1)=\textrm{ord}(p_2)=15$ and $\textrm{ord}(p_3)=5$. Let $P_1,P_2,P_3$ be the
  respective companion
  matrices. 
  One verifies that $\langle P_1\rangle$ and $\langle P_2\rangle$ are
  conjugate to each other
  but $\langle P_3 \rangle$ is not conjugate to them.
\end{example}

One can describe the action of an irreducible matrix group via the
Galois extension field isomorphism.
\begin{theorem}\label{compmult}
  Let $p(x)$ be an irreducible polynomial over $\mathbb{F}_{q}$ of
  degree $n$ and $P$ its companion matrix.  Furthermore let $\alpha
  \in \mathbb{F}_{q^n} $ be a root of $p(x)$ and $\phi$ be the
  canonical homomorphism
  \begin{align*}
    \phi: \mathbb{F}_q^n &\longrightarrow \mathbb{F}_{q^n} \cong \F_{q}[\alpha] \\
    (v_1,\dots,v_n) &\longmapsto \sum_{i=1}^n v_i \alpha^{i-1} .
  \end{align*}
  Then the following diagram commutes (for $v \in \mathbb{F}_q^n $):
  \[
   \begin{diagram}
  v & \rTo^{\cdot P} & vP \\
  \dTo^{\phi} & & \dTo_{\phi}\\
  v' & \rTo_{\cdot\alpha} & v'\alpha
  \end{diagram}
  \]
\end{theorem}

If $P$ is a companion matrix of a primitive polynomial the group
generated by $P$ is also known as a \textit{Singer group}. This
notation is used e.g. by Kohnert et al. in their network code
construction \cite{el10,ko08p}. Elsewhere $P$ is called \textit{Singer
  cycle} or \textit{cyclic projectivity}  (e.g. in \cite{hi98}).

\section{Irreducible Cyclic Orbit Codes}\label{icoc}

The irreducible cyclic subgroups of $\GL_n$ are exactly the groups generated
by the companion matrices of the irreducible polynomials of degree $n$
and their conjugates. Moreover, all groups generated by companion matrices of irreducible
polynomials of the same order are conjugate.

The following theorem shows that is sufficient to characterize the
orbits of cyclic groups generated by companion matrices of irreducible
polynomials of degree $n$.

\begin{theorem}
  Let $G$ be an irreducible matrix, $\mathfrak{G} =\langle G \rangle$
  and $\mathfrak{H} =\langle S^{-1}GS \rangle$ for an $S\in \GL_n$.
  Moreover, let $\mathcal{U} \in \G$ and $\mathcal{V} :=
  \mathcal{U}S$. Then the orbit codes
  \[\mathcal{C}:=\{\mathcal{U}A \mid A \in \mathfrak{G} \} \textnormal{
    and } \mathcal{C'}:=\{\mathcal{V}B \mid B \in \mathfrak{H} \}\] have
  the same cardinality and minimum distance.
\end{theorem}
\begin{proof}
  Trivially the cardinality of both codes is the same. It remains to
  be shown that the same holds for the minimum distance.

 Since
\[\Vvs (S^{-1}GS)^{i}=\Vvs S^{-1}G^{i}S= \Uvs S S^{-1}G^{i}S=\Uvs G^{i}S\]
 and the subspace distance is invariant under $\GL_n$-action, it holds that
  \[d_S(\mathcal{U}, \mathcal{U}G^i)=d_S(\mathcal{V},\mathcal{U}G^i
  S)\] hence the minimum distances of the codes defined by
  $\mathfrak{G}$ and by $\mathfrak{H}$ are equal.   \qed 
\end{proof}

\subsection{Primitive Generator}\label{prim}

Let $\alpha$ be a primitive element of $\F_{q^n}$ and assume $k|n$ and
$c:=\frac{q^n-1}{q^k-1}$. Naturally, the subfield $\F_{q^{k}}\leq \F_{q^{n}}$ is also an $\F_{q}$-subspace of $\F_{q^{n}}$. On the other hand, $\F_{q^{k}} = \{\alpha^{ic} \mid i=0,...,q^{k}-2\}\cup\{0\}$.

\begin{lemma}
  For every $\beta\in\F_{q^n}$ the set
$$
\beta\cdot\F_{q^k}=\{\beta\alpha^{ic} \mid i=0,...,q^{k}-2\}\cup\{0\}
$$
defines an $\F_q$-subspace of dimension $k$.
\end{lemma}
\begin{proof}
Since $\F_{q^{k}}$ is a subspace of dimension $k$ and
\begin{align*}
\varphi_\beta:\ \F_{q^n} & \longrightarrow \F_{q^n}\\ 
u &\longmapsto \beta u
\end{align*}
is an $\F_q$-linear isomorphism, it follows that 
$\varphi_\beta(\F_{q^k})=\beta\cdot\F_{q^k}$ is an $\F_q$-subspace of dimension $k$. \qed
\end{proof}

\begin{theorem}
 The set
$$
\mathcal{S}=\left\{ \alpha^i\cdot\F_{q^k}\mid i=0,\ldots,c-1\right\}
$$
is a spread of $\F_{q^{n}}$ and thus defines a spread code in $\G$. 
\end{theorem}
\begin{proof}
  By a simple counting argument it is enough to show that the subspace
  $\alpha^i\cdot\F_{q^k}$ and $\alpha^j\cdot\F_{q^k}$ have only
  trivial intersection whenever $0\leq i<j\leq c-1$.  For this assume
  that there are field elements $c_i,c_j\in \F_{q^k}$, such that
$$
v=\alpha^i c_i=\alpha^j c_j\in \alpha^i\cdot\F_{q^k}\cap
\alpha^j\cdot\F_{q^k}.
$$
If $v\neq 0$ then $\alpha^{i-j}=c_j c_i^{-1}\in \F_{q^k}.$ But this
means $i-j\equiv 0 \mod c$ and $ \alpha^i\cdot\F_{q^k}=
\alpha^j\cdot\F_{q^k}$, which contradicts the assumption.  It follows that $\mathcal{S}$ is a spread.
 \qed
\end{proof}

We now translate this result into a matrix setting.
For this let $\phi$  denote the canonical homomorphism as
defined in Theorem \ref{compmult}.

\begin{corollary}\label{spread}
  Assume $k|n$.  Then there is a subspace 
 $\mathcal{U} \in \G$ such that the cyclic orbit code obtained by
  the  group action of a 
 a primitive companion matrix is a code with
  minimum distance $2k$ and cardinality $\frac{q^n-1}{q^k-1}$. Hence
  this  irreducible cyclic orbit code is a spread code.
\end{corollary}
\begin{proof}
  In the previous theorem represent $\F_{q^k}\subset \F_{q^n}$ as the
  row space of a $k\times n$ matrix $U$ over $\F_q$ and, using the
  same basis over $\F_q$, represent the primitive element $\alpha$
  with its respective companion matrix $P$. Then the orbit code $\Cvs = \rs(U) \langle
  P\rangle$ has all the desired properties.  \qed 
\end{proof}

\begin{example}
  Over the binary field let $p(x):=x^6+x+1$ be primitive, $\alpha$ a root
  of $p(x)$ and $P$ its companion matrix.
  \begin{enumerate}
  \item For the 3-dimensional spread compute $c=\frac{63}{7}=9$ and
    construct a basis for the starting point of the orbit:
    \begin{align*}
      u_1&=\phi^{-1} (\alpha^0)=\phi^{-1} (1)=(100000)\\
      u_2&=\phi^{-1} (\alpha^c)=\phi^{-1}(\alpha^9)=\phi^{-1}(\alpha^4+\alpha^3)=(000110)\\
      u_3&=\phi^{-1} (\alpha^{2c})=\phi^{-1}(\alpha^{18})=\phi^{-1}(
      \alpha^3+\alpha^2+\alpha+1)= (111100)
    \end{align*}
    The starting point is
    \[ \mathcal{U}=\rs\left[\begin{array}{cccccc} 1&0&0&0&0&0\\
        0&0&0&1&1&0\\ 1&1&1&1&0&0 \end{array}\right] =
    \rs\left[\begin{array}{cccccc} 1&0&0&0&0&0\\ 0&1&1&0&1&0 \\
        0&0&0&1&1&0 \end{array}\right] \] and the orbit of the group
    generated by $P$ on $\mathcal{U}$ is a spread code.
  \item For the 2-dimensional spread compute $c=\frac{63}{3}=21$ and
    construct the starting point
    \begin{align*}
      u_1&=\phi^{-1} (\alpha^0)=\phi^{-1}(1)=(100000)\\
      u_2&=\phi^{-1} (\alpha^c)=\phi^{-1}(\alpha^{21})=\phi^{-1}(\alpha^2+\alpha+1)=
      (111000)
    \end{align*}
    The starting point is
    \[\mathcal{U}=\rs\left[\begin{array}{cccccc} 1&0&0&0&0&0\\
        1&1&1&0&0&0 \end{array}\right]=\rs\left[\begin{array}{cccccc}
        1&0&0&0&0&0\\ 0&1&1&0&0&0 \end{array}\right]\]
    and the orbit of the group generated by $P$ is a spread code.
  \end{enumerate}
\end{example}

The following fact is a generalization of Lemma 1 from \cite{ko08p}.

\begin{theorem}\label{thmprim}
Assume $\mathcal{U}=\{0,u_1,\dots,u_{q^k-1}\} \in \G$, 
  \[\phi(u_i)=\alpha^{b_i} \hspace{0.7cm} \forall i=1,\dots,q^k-1\]
  and $d$ be minimal such that any element of the set
  \[D:=\{b_{m}- b_{l} \mod q^n-1 \mid l,m \in \mathbb{Z}_{q^k-1}, l\neq
  m\}\] has multiplicity less than or equal to $q^d-1$, i.\,e. a quotient of two
  elements in the field representation appears at most $q^d-1$ times
  in the set of all pairwise quotients. If $d<k$ then the orbit of the group
  generated by the companion matrix $P$ of $p(x)$ on $\mathcal{U}$ is
  an orbit code of cardinality $q^{n}-1$ and minimum distance $2k-2d$.
\end{theorem}
\begin{proof}
  In field representation the elements of the orbit code are:
  \begin{align*}
  C_{0}=&\{\alpha^{b_{1}}, \alpha^{b_{2}},...,\alpha^{b_{q^{k}-1}}\}\cup \{0\}\\
  C_{1}=&\{\alpha^{b_{1}+1}, \alpha^{b_{2}+1},...,\alpha^{b_{q^{k}-1}+1}\}\cup \{0\}\\
  \vdots &\\
  C_{q^{n}-2}=&\{\alpha^{b_{1}+q^{n}-2}, ...,\alpha^{b_{q^{k}-1}+q^{n}-2}\}\cup \{0\}
  \end{align*}
  Assume without loss of generality that the first $q^{d}-1$ elements
  of $C_{h}$ are equal to the last elements of $C_{j}$:
  \begin{align*}
  \alpha^{b_{1}+h} = \alpha^{b_{q^{k}-q^{d}}+j} &\iff b_{1}+h\equiv b_{q^{k}-q^{d}}+j \mod q^{n}-1\\
  &\quad \vdots \\
  \alpha^{b_{q^{d}}+h} = \alpha^{b_{q^{k}-1}+j} &\iff b_{q^{d}}+h\equiv b_{q^{k}-1}+j \mod q^{n}-1
  \end{align*}
  To have another element in common there have to exist $y$ and $z$ such that
  \[ b_{q^{k}-q^{d}}-b_1 \equiv b_z -b_y \mod q^n-1 .\] But by
  condition there are up to $q^d-1$ solutions in $(y, z)$ for this
  equation, including the ones from above. Thus the intersection of
  $C_i$ and $C_j$ has at most $q^d-1$ non-zero elements. On the other
  hand, one can always find $h\neq j$ such that there are $q^d-1$
  solutions to
  \[ b_y +h \equiv b_z +j \mod q^n-1 , \] hence,  the minimum distance is
  exactly $2k-2d$.  
  \qed
\end{proof}

\begin{proposition}\label{d=k}
In the setting of before, if $d=k$, one gets orbit elements with full
intersection which means they are the same vector space. 
\begin{enumerate}
\item Let $m(a)$ denote the multiplicity of an element in the
  respective set and $D':=D\setminus \{a\in D \mid m(a)=q^{k}-1\}$.
  Then the minimum distance of the code is $2k - 2d'$ where
  $d':=\log_{q}(\max\{m(a) \mid a\in D'\})$.
\item Let $m$ be the least element of $D$ of multiplicity $q^{k}-1$.
  Then the cardinality of the code is $m-1$.
\end{enumerate}
\end{proposition}
\begin{proof}
\begin{enumerate}
\item Since the minimum distance of the code is only taken between
  distinct vector spaces, one has to consider the largest intersection
  of two elements whose dimension is less than $k$.
\item Since 
\[\Uvs P^{m} = \Uvs \implies \Uvs P^{lm} = \Uvs \quad \forall l \in \mathbb{N}\]
and the elements of $D$ are taken modulo the order of $P$, one has to
choose the minimal element of the set $\{a \in D \mid m(a)=q^{k}-1\}$
for the number of distinct vector spaces in the orbit.
\end{enumerate}
\qed
\end{proof}

\subsection{Non-Primitive Generator}

\begin{theorem}\label{non1}
  Let $P$ be an irreducible non-primitive companion matrix,
  $\mathfrak{G}$ the group generated by it and denote by
  $v\mathfrak{G}$ and $\mathcal{U}\mathfrak{G}$ the orbits of
  $\mathfrak{G}$ on $v\in \F_{q}^{n}$ and $\mathcal{U}\in \G$,
  respectively.  If $\mathcal{U} \in \G$ such that
  \[v\neq w \implies v\mathfrak{G} \neq w\mathfrak{G}  \quad \forall \: v,w \in \mathcal{U} , \]
  then $\mathcal{U}\mathfrak{G}$ is an orbit code with minimum distance $2k$ and
  cardinality $\mathrm{ord}(P)$.
\end{theorem}
\begin{proof}
The cardinality follows from the fact that each element of $\mathcal{U}$ has its own orbit of
  length $\mathrm{ord}(P)$. Moreover, no code words intersect
  non-trivially, hence the minimum distance is $2k$.   \qed \end{proof}

Note that, if the order of $P$ is equal to $\frac{q^{n}-1}{q^{k}-1}$,
these codes are again spread codes.

\begin{example}
  Over the binary field let $p(x)=x^4+x^3+x^2+x+1$, $\alpha$ a root of
  $p(x)$ and $P$ its companion matrix. Then $\F_{2^4}\setminus \{0\}$
  is partitioned into
  \[\{\alpha^i \mid i=0, \dots , 4\} \cup \{\alpha^i(\alpha+1) \mid i=0,
  \dots , 4\} \cup \{\alpha^i(\alpha^2+1) \mid i=0, \dots , 4\} .\] Choose
  \begin{align*}
    u_1=&\phi^{-1}(1)=\phi^{-1}(\alpha^0)=(1000)\\
    u_2=&\phi^{-1}(\alpha^3+\alpha^2)=\phi^{-1}(\alpha^2(\alpha+1))=(0011)\\
    u_3=&u_1+u_2=\phi^{-1}(\alpha^3+\alpha^2+1)
    =\phi^{-1}(\alpha^4(\alpha^2+1))=(1011)
  \end{align*}
  such that each $u_i$ is in a different orbit of $\langle P\rangle$
  and $\mathcal{U}=\{0,u_1,u_2,u_3\}$ is a vector space.
  Then the orbit of $\langle P\rangle$ on $\mathcal{U}$ has minimum
  distance $4$ and cardinality $5$, hence it is a spread code.
\end{example}

\begin{proposition}
  Let $P$ and $\mathfrak{G}$ be as before and $\mathcal{U}=\{0,v_{1},
  \dots, v_{q^{k}-1}\} \in \G$.  Let $l=\frac{q^{n}-1}{\textrm{ord}(P)}$ and $O_1,...,O_l$ be the different
  orbits of $\mathfrak{G}$ in $\F_{q}^{n}$. Assume that $m<q^k-1$
  elements of $\mathcal{U}$ are in the same orbit, say $O_1$, and all
  other elements are in different orbits each, i.e.
 \[v_{i}\mathfrak{G} = v_{j}\mathfrak{G} = O_{1}\quad \forall \: i,j \leq m , \]
  \[v_{i}\neq v_{j} \implies v_{i}\mathfrak{G} \neq v_{j}\mathfrak{G}  \quad \forall \: i,j\geq m . \]
  Apply the theory of Section \ref{prim} to the orbit
  $O_1$ and find $d_1$ fulfilling the conditions of Theorem
  \ref{thmprim}. Then the orbit of $\mathfrak{G}$ on $\mathcal{U}$ is
  a code of length $\mathrm{ord}(P)$ and minimum distance $2k-2d_1$.
\end{proposition}
\begin{proof}
  \begin{enumerate}
  \item Since there is at least one orbit $O_i$ that contains exactly
    one element of $\mathcal{U}$, each element of $O_i$ is in exactly
    one code word. Hence the cardinality of the code is
    $\mathrm{ord}(\mathfrak{G})=\mathrm{ord}(P)$.
  \item In analogy to Theorem \ref{non1} the only possible
    intersection is inside $O_1$, which can be found according to the
    theory of primitive cyclic orbit codes.
  \end{enumerate}
 \qed
\end{proof}

We generalize these results to any possible starting point $\in \G$:

\begin{theorem}
  Let $P, \mathfrak{G}, \mathcal{U}$ and the orbits $O_1,...,O_l$ be
  as before. Assume that $m_i$ elements of $\mathcal{U}$ are in the
  same orbit $O_i$ ($i=1,\dots, l$). Apply the theory of Section
  \ref{prim} to each orbit $O_i$ and find the corresponding $d_i$ from
  Theorem \ref{thmprim}.  Then the following cases can occur:
  \begin{enumerate}
  \item No intersections of two different orbits coincide. Define $d_{\max}:=
    \max_{i}d_i$.  Then the orbit of $\mathfrak{G}$ on $\mathcal{U}$
    is a code of length $\mathrm{ord}(P)$ and minimum distance
    $2k-2d_{max}$.
  \item Some intersections coincide among some orbits. Then the corresponding $d_i$'s
    add up and the maximum of these is the maximal intersection number $d_{\max}$. 
  \end{enumerate}
  Mathematically formulated: Assume $b_{(i,1)},\dots ,
  b_{(i,\mathrm{ord} (P)-1)}$ are the exponents of the field
  representation of the non-zero elements of $\mathcal{U}$ on $O_i$.
  For $i=1,\dots,l$ define
\[
a_{(i,\mu,\lambda)} := b_{(i,\mu)} - b_{(i,\lambda)}, 
\]
\[D_i:=\{a_{(i,\mu,\lambda)}\mid \mu,\lambda \in \{1,\dots, \mathrm{ord}(P)-1\}\},\]
and the difference set
\[D:=\bigcup_{i=1}^l D_i .\] Denote by $m(a)$ the multiplicity of an
element $a$ in $D$ and $d_{\max} := \log_q(\max \{m(a) \mid a\in
D\}+1)$.  Then the orbit of $\mathfrak{G}$ on $\mathcal{U}$ is a code
of length $\mathrm{ord}(P)$ and minimum distance $2k-2d_{\max}$.
\end{theorem}

Again note that, in the case that the minimum distance of the code is $0$, one has
  double elements in the orbit. Then Proposition \ref{d=k} still holds.

\begin{remark}
The theorems about the minimum distance can also be used for the construction of orbit codes with a prescribed minimum distance. For this construct the initial point of the orbit by iteratively joining elements $\alpha^{i} \in \F_{q^{k}}$ such that the linear span of the union fulfills the condition on the differences of the exponents. 
\end{remark}

  \section{Pl\"ucker Embedding}\label{PC}

  For the remainder of this paper let $p(x) = \sum_{i=0}^{n} p_{i}x^{i}
  \in \F_{q}[x]$ be an irreducible polynomial of degree $n$ and
  $\alpha$ a root of it. The companion matrix of $p(x)$ is denoted by
  $P$. $\F_{q}^{\times} := \F_{q}\setminus \{0\}$ is the set of all
  invertible elements of $\F_{q}$.

  Moreover, let $A\in \Mat_{m\times n}$ such that $m,n \geq k$. Denote
  by $A_{i_1,...,i_k}[j_1,...,j_k]$ the $k\times k$ submatrix of $A$
  defined by rows $i_1,...,i_k$ and columns $j_1,...,j_k$ and
  $A[j_1,...,j_k]$ denotes the submatrix of $A$ with the complete
  columns $j_1,...,j_k$.

\begin{definition}
  We define the following operation on $ \Lambda^{k}(\F_{q}[\alpha])
  \cong \Lambda^{k}(\F_q^{n})$ :
  \begin{align*}
    * : \Lambda^{k}(\F_{q}[\alpha])\times \F_q[\alpha]\setminus\{0\} 
    &\longrightarrow \Lambda^{k}(\F_{q}[\alpha])\\
    ((v_{1}\wedge ... \wedge v_{2}), \beta) &\longmapsto (v_{1}\wedge
    ... \wedge v_{2}) * \beta := (v_{1}\beta \wedge ... \wedge
    v_{k}\beta) .
  \end{align*}
  This is a group action since $((v_{1}\wedge ... \wedge v_{2}) *
  \beta) * \gamma = (v_{1}\wedge ... \wedge v_{2}) * (\beta\gamma)$.
\end{definition}

\begin{theorem}
  The following maps are (isomorphic) embeddings of the Grassmannian:
  \begin{align*}
    \varphi : \G &\longrightarrow \mathbb{P}^{\binom{n}{k}-1} \\
    \rs (U) &\longmapsto [\det(U[1,...,k]) : \det(U[1,...,k-1,k+1]) : ... : 
    \det(U[n-k+1,...,n])]\\ \vspace{0.5cm}\\
    \varphi' : \G &\longrightarrow\mathbb{P}(\Lambda^{k}(\F_{q}[\alpha]) )\\
    \rs (U) &\longmapsto (\phi(U_{1}) \wedge ... \wedge \phi(U_{k}))
    *\F_{q}^{\times}
  \end{align*}
  where $\phi : \F_{q}^{n} \rightarrow \F_{q^{n}}$ denotes the
  standard vector space isomorphism. 
\end{theorem}
\begin{proof}
First we show that $\varphi$ is an embedding. For this assume 
that $U,V$ are two full-rank $k\times n$ matrices such that $\rs (U)=\rs (V)$. It follows that there is 
an $S\in \GL_{k}$ with $V=SU$. The two vectors
 $$
[\det(U[1,...,k]), \det(U[1,...,k-1,k+1]),\ldots , 
    \det(U[n-k+1,...,n])]
$$
and 
 $$
[\det(V[1,...,k]), \det(V[1,...,k-1,k+1]),\ldots , 
    \det(V[n-k+1,...,n])]
$$
differ hence only by the non-zero factor $\det S$. As elements 
of the projective space $\mathbb{P}^{\binom{n}{k}-1}$ they are
thus the same and the map is well defined.

Assume now that  $\rs (U)\neq \rs (V)$. Without loss of generality
we can assume that both $U$ and $V$ are in reduced row echelon form, where
 the  forms are necessarily different. Observe that all non-zero entries
of $U$ can also be written, up to sign, as $\det (U[i_1,...,i_k])$. It follows that  $\varphi(U)$ is 
different from $\varphi(V)$.

Next we show that the map $\psi : \varphi'(\G) \rightarrow \varphi(\G)$, defined as follows, is an isomorphism.
  \begin{align*}
    (\phi(U_{1}) \wedge ... \wedge \phi(U_{k})) * \F_{q}^{\times} 
    &= (\sum_{i=0}^{n-1} \lambda_{1i} \alpha^{i} \wedge ... 
    \wedge \sum_{i=0}^{n-1} \lambda_{ki} \alpha^{i}) * \F_{q}^{\times} \\
    &= \sum_{0\leq i_{1},...,i_{k} < n} (\lambda_{1i_{1}} \alpha^{i_{1}} 
     \wedge ... \wedge \lambda_{ki_{k}} \alpha^{i_{k}}) * \F_{q}^{\times} \\
    &=  \sum_{0\leq i_{1},...,i_{k} < n} {\lambda_{1i_{1}} ...  \lambda_{ki_{k}}}
     (\alpha^{i_{1}} \wedge ... \wedge \alpha^{i_{k}})* \F_{q}^{\times} \\
    &= \sum_{0\leq i_{1}<...<i_{k} < n} \mu_{i_{1}, ... , i_{k}} (\alpha^{i_{1}} 
      \wedge ... \wedge \alpha^{i_{k}})* \F_{q}^{\times} \\
    &\longmapsto [\mu_{0, ... , k-1} : ... : \mu_{n-k, ... , n-1} ]
  \end{align*}
  where $\lambda_{jl} \in \F_{q}$ for all $j\in \{1,...,k\}, l\in
  \{0,...,n-1\}$ and $\mu_{i_{1},...,i_{k}} := \sum_{\sigma \in S_{k}}
  (-1)^{\sigma} \lambda_{1\sigma(i_{1})} ...  \lambda_{k\sigma(i_{k})}
  \in \F_{q}$.

  Since $\psi$ is an isomorphism and $\varphi' = \psi^{-1} \circ
  \varphi$, it follows that $\varphi'$ is an embedding as well.  \qed
\end{proof}

\begin{remark}
The map $\varphi$ is called the \emph{Pl\"ucker embedding} of the Grassmannian $\G$.
The projective coordinates 
$$ 
[\det(U[1,...,k]) : ... : \det(U[n-k+1,...,n])] =  
\F_{q}^{\times} (\det(U[1,...,k]) , ... , \det(U[n-k+1,...,n])).
$$ 
are often referred to as the 
\emph{Pl\"ucker coordinates} of $\rs(U)$.
\end{remark}

\begin{theorem}
The following diagram commutes:
\begin{align*}
\begin{diagram}
  \Uvs	&\rTo^{\cdot P}	&\Uvs P\\
  \dTo^{\varphi'} &	&\dTo_{\varphi'}\\
  \sum \mu_{i_{1},\dots,i_{k}} (\alpha^{i_{1}}\wedge ... \wedge
  \alpha^{i_{k}})* \F_{q}^{\times} &\rTo^{\phantom{\text{aaa}} *\alpha
    \phantom{\text{aaa}}} & \sum \mu_{i_{1},\dots,i_{k}}
  (\alpha^{i_{1}+1}\wedge ... \wedge \alpha^{i_{k}+1})*
  \F_{q}^{\times}
\end{diagram}&&
\end{align*}
Hence, an irreducible cyclic orbit code $\Cvs = \{\Uvs
P^i \mid i=0,...,\mathrm{ord}(P)-1\}$ has a corresponding ``Pl\"ucker
orbit'':
\[
\varphi'(\Cvs) = \{\varphi'(\Uvs) * \alpha^i \mid i=0,...,\mathrm{ord}(\alpha)-1\}
 = \varphi'(\Uvs) * \langle \alpha \rangle
\]
\end{theorem}

\begin{proof}
\begin{align*}
\varphi'(\Uvs P) &= \F_{q}^{\times} \cdot (\phi(U_{1}P) 
\wedge  ...\wedge  \phi(U_{k}P)) = \F_{q}^{\times} \cdot 
(\phi(U_{1})\alpha \wedge  ...\wedge  \phi(U_{k})\alpha) \\
&= \F_{q}^{\times} \cdot (\phi(U_{1}) \wedge  ...\wedge  \phi(U_{k})) *\alpha 
\end{align*}
\qed
\end{proof}

\begin{example}
  Over $\F_2$ let $p(x)=x^4+x+1$ and $\Uvs \in \mathcal{G}_2(2,4)$
  such that $\phi(\Uvs)=\{0, 1, \alpha+\alpha^2, 1+\alpha+\alpha^2\}$,
  i.e.
  \[\Uvs = \rs\left[\begin{array}{cccc}1 & 0&0&0\\ 0&1&1&0
    \end{array}\right]  .\]
  Then $ \varphi'(\Uvs) = (1 \wedge \alpha+\alpha^2) = (1 \wedge
  \alpha) + (1\wedge \alpha^{2})$ and $\varphi(\Uvs) = [\mu_{0,1}:
  \mu_{0,2}: \mu_{0,3} : \mu_{1,2}: \mu_{1,3} :\mu_{2,3}] =
  [1:1:0:0:0:0] $. The elements of the Pl\"ucker orbit
  $\varphi'(\Uvs)*\langle \alpha \rangle$ are
  \begin{align*}
    & (1 \wedge \alpha+\alpha^2) = (1 \wedge \alpha) + (1\wedge \alpha^2) ,\\
    & (\alpha \wedge \alpha^2+\alpha^3) = (\alpha \wedge \alpha^2) + (\alpha \wedge \alpha^3) ,\\
    & (\alpha^2 \wedge \alpha^3+\alpha^4) = 
    (\alpha^2 \wedge 1+\alpha+ \alpha^3) = (\alpha^2 \wedge 1)
     +(\alpha^2 \wedge \alpha) + (\alpha^2 \wedge \alpha^3) ,\\
    & (\alpha^3 \wedge \alpha+\alpha^2+\alpha^4) =  (\alpha^3 \wedge 1+\alpha^2) 
    = (\alpha^3 \wedge 1) + (\alpha^3 \wedge \alpha^2) ,\\
    & (\alpha^4 \wedge \alpha+\alpha^3) = (1+\alpha \wedge
    \alpha+\alpha^3)= (1 \wedge \alpha)+ (1 \wedge \alpha^3) + (\alpha
    \wedge \alpha^3),
  \end{align*}
  and $ (\alpha+\alpha^2 \wedge \alpha^2+\alpha^4) = (\alpha+\alpha^2
  \wedge 1+\alpha+\alpha^2) = (\alpha+\alpha^2 \wedge 1)= (1 \wedge
  \alpha+\alpha^2)$ over $\F_2$. The corresponding Pl\"ucker
  coordinates are
  \begin{align*}
    & [ 1:1:0:0:0:0] ,\\
    & [ 0:0:0:1:1:0]  ,\\
    & [ 0:1:0:1:0:1] ,\\
    & [ 0:0:1:0:0:1] ,\\
    & [ 1:0:1:0:1:0] .
  \end{align*}
  The respective subspace code is the spread code defined by $x^4+x+1$
  according to Section \ref{prim}.
\end{example}

In the following we describe the balls of radius $t$
(with respect to the subspace distance) around some $\Uvs \in \G$ with
the help of the Pl\"ucker coordinates. An algebraic description of the balls of
radius $t$ is potentially important if one is interested in an algebraic
decoding algorithm for constant dimension codes. For example, a list decoding algorithm
would compute all code words inside some
ball around a received message word.

The main result shows that the balls of radius $t$ have the structure
of Schubert varieties~\cite[p. 316]{ho52}. In order to establish this
result we introduce the following partial order:

\begin{definition}
  Consider the set $\binom{[n]}{k} := \{(i_1,...,i_k) \mid i_l \in
  \mathbb{Z}_n \;\forall l\}$ and define the partial order
  \[
  \mathbf{i}:=(i_{1},...,i_k) > (j_{1},...,j_{k})=:\mathbf{j} \iff \exists N\in
  \mathbb{N}_{\geq 0} : i_{l}=j_{l} \;\forall l<N \textnormal{ and }
  i_{N} > j_{N} .\]
\end{definition}

It is easy to compute the balls around a vector space in the following
special case.

\begin{proposition}
  Denote the balls of radius $2t$ centered at $\Uvs$ in $\G$ by
  $B_{2t}(\Uvs)$ and define $\Uvs_{0}:=\rs
  [\begin{array}{cc}I_{k\times k} &0_{k\times n-k} \end{array}]$. Then
  \begin{align*}
    B_{2t}(\Uvs_{0}) = \{&\Vvs \in \G \mid \varphi'(\Vvs) =
    \det(\Vvs[i_1,...,i_{k}]) = 0 \;\forall (i_{1},...,i_{k}) \not
    \leq (t+1,...,k,n-t+1,...,n) \}
  \end{align*}
\end{proposition}

\begin{proof}
  For $\Vvs$ to be inside the ball it has to hold that
  \begin{align*}
    d_{S}(\Uvs_{0}, \Vvs) &\leq 2t \\
    \iff 2k - 2\dim(\Uvs_{0} \cap \Vvs) & \leq 2 t \\
    \iff \dim(\Uvs_{0} \cap \Vvs) & \geq k-t ,
  \end{align*}
  i.e. $k-t$ many of the unit vectors $e_{1},..., e_{k}$ have to be
  elements of $\Vvs$. Since $\phi(e_{j})= \alpha^{j-1}$, it follows
  that $\varphi'(\Vvs)$ has to fulfill
  \[\mu_{i_1,...,i_{k}} = 0 \textnormal{ if } (i_{1},...,i_{k}) \not
  \leq (t+1,...,k,n-t+1,...,n) .
  \]
  \qed
\end{proof}

The proposition shows that $B_{2t}(\Uvs_0)$ is described in the 
Pl\"ucker space $ \mathbb{P}^{\binom{n}{k}-1}$ as a point in the
Grassmannian together with linear constraints on the Pl\"ucker coordinates.

\begin{example}
  In $\mathcal{G}_{2}(2,4)$ we have
  \[\Uvs_{0} = \rs\left[\begin{array}{cccc}1 & 0&0&0\\ 0&1&0&0
    \end{array}\right]  \]
  and the elements of distance $2$ (i.e. $t=1$) are
  \begin{align*}
    B_{2}(\Uvs_{0}) = \{&\Vvs \in \mathcal{G}_{2}(2,4) 
    \mid \det(\Vvs[i_1, i_{2}]) = 0  \;\forall (i_{1}, i_{2}) \not \leq  (2,4)  \} \\
    = \{&\Vvs \in \mathcal{G}_{2}(2,4) \mid \det(\Vvs[3,4]) = 0 \} .
  \end{align*}
\end{example}

Next we derive the equations for a ball $B_{2t}(\Uvs)$ around an arbitrary 
subspace $\Uvs\in\G$. For this  assume that $\Uvs = \Uvs_0 G$ for some
$G\in \GL_n$. A direct computation shows that
\[
B_{2t}(\Uvs) = B_{2t}(\Uvs_0 G) =B_{2t}(\Uvs_0) G .
\]

The transformation by $G$ transforms the linear equations
$\det(\Vvs[i_1,...,i_{k}]) = 0 \;\forall (i_{1},...,i_{k}) \not \leq
(t+1,...,k,n-t+1,...,n)$ into a new set of linear equations in the
Pl\"ucker coordinates. Instead of deriving these equations in an
explicit manner we will show instead that the ball $B_{2t}(\Uvs)$
describes a Schubert variety. Then we will show that
the equations defining the ball consist of the defining equations
of the Grassmann variety together with a set of linear equations
describing the Schubert variety.

\begin{definition}
A flag ${\cal F}$ is a sequence of nested subspaces
\begin{equation*}
\{ 0\}\subset V_{1}\subset V_{2}\subset\ldots \subset
V_{n}= \F_{q}^n
\end{equation*}
where we assume that $\dim V_{i}=i$ for $ i=1,\dots ,n$.
\end{definition}

\begin{definition}
Consider the multi-index $\mathbf{i}= (i_{1},\ldots ,i_{k})$
such that $1\leq i_{1} <\ldots <i_{k} \leq n $. Then
 $$
S(\mathbf{i};{\cal F})  :=
\{ \Vvs\in \G \mid \dim (\Vvs\cap V_{i_{s}}) \geq s\} 
$$
is called a \emph{Schubert variety}.
\end{definition}

One observes that $B_{2t}(\Uvs)=\{\Vvs \in \G \mid
\dim(\Uvs \cap \Vvs) \geq k-t\}$
 is nothing else than
a special Schubert variety. Indeed, we can simply choose a flag ${\cal F}$
having the property that $V_k=\Uvs$ in conjunction with the multi-index
$\mathbf{i}=(t+1,...,k,n-t+1,...,n)$. 

Next we describe the defining equations inside the Pl\"ucker space $
\mathbb{P}^{\binom{n}{k}-1}$. For this introduce a basis $\{
e_{1},\ldots ,e_{n}\} $ of $\F_{q}^n$ which is compatible with the flag
${\cal F}$, i.e. $\mathrm{span}\{
e_{1},\ldots ,e_{i}\}=V_i$ for $ i=1,\ldots,n.$

The basis $\{ e_{1},\ldots ,e_{n}\} $ induces the basis
$$
\{ e_{i_{1}}\wedge \ldots \wedge e_{i_{k}} \mid 1\leq i_{1}
<\dots <i_{k} \leq n\}
$$
of $\Lambda^{k}(\F_{q}^n)$. If $x\in \Lambda^{k}(\F_{q}^n)$, denote by $x_\mathbf{i}$ its coordinate with 
regard to the basis vector $e_{i_{1}}\wedge \ldots \wedge e_{i_{k}}$.
The defining equations of the Schubert variety $S(\mathbf{i};{\cal F})$
are then given by
\begin{equation*}
S(\mathbf{i};{\cal F})=\{ x\in \G\mid x_\mathbf{j}=0,\ \ \forall \mathbf{j}\not
  \leq\mathbf{i}\} .
\end{equation*} 

An elementary proof of the fact that these linear equations 
together with the defining equations of the Grassmannian $\G$ indeed describe 
the  Schubert variety $S(\mathbf{i};{\cal F})$ can be found in~\cite[Chapter XIV]{ho52}.

For coding theory it is important to note that we have explicit equations
describing Schubert varieties in general and balls of radius $t$ in particular.
If a constant dimension network code is given by explicit equations,
one would immediately have a description of all code words which are closer than a given distance to some received subspace.




\section{Conclusion}\label{conclusion}

We listed all possible irreducible cyclic orbit codes and showed that
it suffices to investigate the groups generated by companion matrices
of irreducible polynomials. Moreover, polynomials of the same degree
and same order generate codes with the same cardinality and minimum
distance. These two properties of the code depend strongly on the
choice of the starting point in the Grassmannian. We showed how one can
deduce the size and distance of an orbit code for a given subgroup of
$\GL_n$ from the starting point $\mathcal{U} \in \G$. For primitive
groups this is quite straight-forward while the non-primitive case is
more difficult.

Subsequently one can use this theory of irreducible cyclic orbit codes
to characterize all cyclic orbit codes. 

Finally we described 
the irreducible cyclic orbit codes within the Pl\"ucker space and showed that 
the orbit structure is preserved. Moreover, we showed how balls around an element of the Grassmann variety can be described using Pl\"ucker coordinates.


\bibliography{/home/a/rosen/Bib/huge.bib}
\bibliographystyle{plain}

\end{document}